\documentclass{amsart}
\usepackage{ mathrsfs, amsthm, amssymb, amsmath}
\usepackage[all]{xy}

\def\ba #1\ea{\begin{align} #1 \end{align}}
\def\bas #1\eas{\begin{align*} #1 \end{align*}}
\newcommand{\R}{\mathbb{R}}
\newcommand{\id}{\mathbb{I}}
\newcommand{\C}{\mathbb{C}}
\newcommand{\h}{\mathcal{H}}

\newcommand{\scrA}{\mathscr{A}}
\newcommand{\Spec}{\textrm{Spec }}
\newcommand{\Hom}{\textrm{Hom}}
\newcommand{\g}{\mathfrak{g}}
\newcommand{\CP}{\C\mathbb{P}}

\newcommand{\dt}{t\frac{\partial}{\partial t}}

\newtheorem{theorem}{Theorem}[section]

\newtheorem{lemma}[theorem]{Lemma}

\newtheorem{remark}{Remark}
\newtheorem{definition}{Definition}

\begin{document}

\markboth{Susama Agarwala}
{Momentum cut-off and dimensional
  regularization}

\title{Geometrically relating momentum cut-off and dimensional
  regularization}

\author{SUSAMA AGARWALA}
\address{ Mathematics Department \\ California Institute of Technology \\
1200 E California Ave. \\
Pasadena, CA 91106}
\email{susama@caltech.edu}

\maketitle


\begin{abstract} The $\beta$ function for a scalar field theory describes the
  dependence of the coupling constant on the renormalization mass
  scale. This dependence is affected by the choice of regularization
  scheme. I explicitly relate the $\beta$ functions of momentum cut-off
  regularization and dimensional regularization on scalar field
  theories by a gauge transformation using the Hopf algebras of the
  Feynman diagrams of the theories.
\end{abstract}


\section{Introduction}
Perturbative quantum field theories (QFTs), when naively calculated,
lead to divergent integrals and undefined quantities. To address this
problem, physicists have developed many regularization and
renormalization schemes to extract finite values from divergent
integrals. Introducing a regularization parameter forces the
quantities in the Lagrangian of the field theory to be dependent on
the energy scale of the calculation. This scale dependence is captured
by a new parameter called the renormalization mass. The theory's
dependence on the renormalization mass is described by a set of
differential equations, called the renormalization group equations, or
RGEs. The simplest of these solves for the dependence of the coupling
constant on the renormalization mass, and gives the $\beta$ function
of the theory. In the perturbative case, the RGEs depend on the
regularization scheme. Different regularization schemes give rise to
different RGEs. Very little is understood about the relationship
between different regularization schemes.

In this paper, I compare regularization schemes with logarithmic
singularities and finite poles to those with only finite
poles. Specifically, I study the relationship between sharp momentum
cut-off regularization and dimensional regularization, and the
associated $\beta$ functions.

Recent literature has emerged geometrically describing this process of
renormalization and regularization for a QFT in which the perturbative
$\beta$ function for dimensional regularization is defined by a
connection on a renormalization bundle in which different
regularization schemes correspond to sections \cite{CK00, CK01, CM06,
  EGK04, EM, thesis1}.  I extend the analysis in these papers in two
directions. First, I extend the Birkhoff decomposition of characters
developed in the literature to include algebras with logarithmic
singularities found in momentum cut-off and related regularization
schemes. Second, I look at a renormalization bundle that admits a
number of different renormalization group actions. Geometrically, the
perturbative $\beta$ function is the vector field generating the group
of one parameter diffeomorphisms induced by the renormalization group
action.  By incorporating multiple renormalization group actions into
one bundle, where each defines a one parameter family of
diffeomorphisms, I can write the perturbative $\beta$ function for
each of these renormalization group actions in terms of a connection
on this extended renormalization bundle. As a result, perturbative
$\beta$ functions for dimensional regularization and momentum cut-off,
which are known to be different for gauge theories beyond the three
loop level can be related in terms of a gauge transformation.

Section 1 recalls some useful facts about Feynman integrals,
dimensional regularization and cut-off regularization. Section 2
constructs the new renormalization bundle and defines the relevant
$\beta$ functions in terms of connections on it.

\section{Momentum Cut-off and Dimensional Regularization}

In this section I consider Feynman integrals of a massive $\phi^4$
theory in $\R^4$ \bas \mathcal{L} = \frac{1}{2} \phi(\Delta + m^2)\phi
+ g \phi(x)^4 \; .\eas The same arguments can be made for other
renormalizable theories. For a graph, $\Gamma$, with $l$ loops, $I$
internal edges, and $J$ external edges with assigned momenta $e_1
\ldots e_J$, the Feynman integral is of the form \ba
\int_{\R^{4l}}\prod_{k=1}^I \frac{1}{f_k(p_i,e_j)^2 + m^2}
\prod_{i=1}^ld^4p_i \label{example}\; ,\ea where the $p_i$ are the
loop momentum assigned to each loop, $f(p_i,e_j)$ is a linear
combination of the loop and external momenta representing the momenta
assigned to each internal leg, and the square refers to a dot product
of the vectors. All calculations in this paper are done in Euclidean
space, all integrals have been Wick rotated. These integrals,
\eqref{example}, are generally divergent as written. The process of
regularization and renormalization extracts physical, finite values
from these divergent integrals. In this section, I recall properties
of dimensional regularization and momentum cut-off regularization. The
latter can be renormalized using the BPHZ algorithm, which iteratively
subtracts off certain Taylor jets in the external momenta. The former
can be renormalized using the BPHZ algorithm with a minimal
subtraction operator replacing the Taylor jet subtraction
operator. Section \ref{BPHZcutoff} recalls that using the minimal
subtraction operator for momentum-cutoff regularization gives a valid
renormalization scheme.

For dimensional regularization, write the integral in \eqref{example}
in spherical coordinates, \bas A(4)^l \int_0^\infty\prod_{k=1}^I
\frac{1}{f_k(p_i,e_j)^2 + m^2} \prod_{i=1}^lp_i^3 dp_i \; ,\eas where
$A(d) = \frac{\Gamma(d)}{(4\pi)^{d/2}}$ is the volume of $S^{d-1}$,
the sphere in $d-1$ dimensions. Dimensional regularization exploits
the fact that the integral above is convergent if taken over $d = 4+
z$, dimensions, with $z$ a complex parameter. Notice that $A(d)$ is
holomorphic in $z$, and does not contribute to the polar structure of
the graph. The dimensionally regularized integral is \bas
\varphi_{dr}(z)(\Gamma) = A(d)^l \int_0^\infty\prod_{k=1}^I
\frac{1}{f_k(p_i,e_j)^2 + m^2} \prod_{i=1}^lp_i^{d-1} dp_i \; .\eas
Put another way, dimensional regularization assigns a holomorphic
function, $A(d)$, times the Mellin transform of each loop integral in
the Feynman integral.  If the original integral is divergent, this
expression has a pole at $d = 4$.

Momentum cut-off regularization multiplies the integrand of the
Feynman integral in polar coordinates by a cutoff function. The
simplest is to impose a sharp cut-off function \bas \chi_\Lambda(p)
=\begin{cases} 1 & \text{if $p\leq \Lambda$,} \\ 0 &\text{if $p >
  \Lambda$.}\end{cases} \; .\eas However, this destroys some nice
analytic properties, and sometimes it is better to examine a smooth
cutoff function. The calculations in the paper are done using sharp
cut-off, but the analysis generalizes to the smooth case. The
philosophy behind cut-off regularization is that physical theories are
only valid in a certain domain. Once the energy scale is large enough,
one doesn't expect the theory to hold. Therefore, one should only
consider energy scales at which the theory is valid.  The momentum
cut-off regularized version of \eqref{example} is \bas
\varphi_{mc}(\Lambda)(\Gamma) = \int_{-\Lambda}^\Lambda\prod_{k=1}^I
\frac{1}{f_k(p_i,e_j)^2 + m^2} \prod_{i=1}^l d^4p_i \; .\eas 

\begin{definition}
A one particle irreducible graph, or a 1PI graph, is a connected graph that
is still connected after the removal of any single (internal)
edge.\end{definition}

Dimensional analysis and power counting arguments show that the only
divergent integrals of renormalizable $\phi^4$ theory in $\R^4$ are
those associated to 1PI graphs with either 2 or 4 external legs ($J
\in \{2, 4\}$) \cite{IZ}.

\begin{definition}
For $\phi^4$ in $\R^4$, the superficial degree of divergence of a 1PI
graph, $\Gamma$ is $\omega(\Gamma) = 4l -2I$. \end{definition}

If $\omega(\Gamma) <0$ then the integral is convergent. If
$\omega(\Gamma) \geq0$, the integral is divergent, \cite{IZ} \S8.1.3.
Dimensional regularization of integrals in renormalizable theories
gives holomorphic functions with finite poles at $z = 0$. Momentum
cutoff regularization with a sharp cut-off function has logarithmic and
polynomial singularities at $\Lambda \rightarrow \infty$, \cite{IZ}
\S 8.2.1. One can impose different cut-off functions to maintain
smoothness or other analytic properties. Then the regulator depends on
the cutoff function.

\subsection{Renormalization group action}
In this section I explicitly show the effect of the renormalization
group on Feynman rules regulated by dimensional regularization and
momentum cutoff regularization. The unregularized Lagrangian for a
field theory is scale invariant, \bas \int_{\R^4} \mathcal{L}(x) d^4x =
\int_{\R^4} \mathcal{L}(tx) d^4(tx) \eas for $t \in \R_{>0}$. Let
$\mathcal{L}(z, x)$ be the regularized Lagrangian density, with
regulator $z$. This scale invariance no longer holds \bas
\int_{\R^4}\mathcal{L}(z, x) d^4x \neq \int_{\R^4} \mathcal{L}(z, tx)
d^4(tx) \;.\eas The action of the renormalization group translates
between the results of different energy scales of the Lagrangian
density for a theory.

Strictly speaking, the renormalization group is actually a torsor $M
\simeq \R_{>0}$. Fixing an energy scale for the theory determines the
identity of the renormalization group. To see the effect of the
renormalization group on Feynman integrals, I write $\varphi_{dr}(m,
e_j, t, z)(\Gamma)$ and $\varphi_{mc}(m, e_j, t, \Lambda)(\Gamma)$ as
the Feynman integral associated to the graph $\Gamma$. Here $m$ is the
mass, $e_j$, the external momenta and $t$ the scale at which the
theory is being evaluated. A Feynman integral taken at the energy
scale $t$ is integrated against the variables $tp_i$, where $p_i$ are
the loop momenta of the graph $\Gamma$. Explicitly, \ba
t^{zl}\varphi_{dr}(m, e_j, 1, z)(\Gamma)=
t^{-\omega(\Gamma)}\varphi_{dr}(tm, te_j, t, z)(\Gamma) 
\label{scalingdr} \\ \varphi_{mc}(m, e_j, 1, \Lambda)(\Gamma)=
t^{-\omega(\Gamma)}\varphi_{mc}(tm, te_j,t, t\Lambda)
(\Gamma) \label{scalingmc} \;. \ea

The extra factor of $t^{zl}$ in the case of dimensional
regularization, is called the t'Hooft mass. The integrals are written
in this manner to keep certain quantities dimensionless. 
On the level of the Lagrangian density, introducing the
energy scale also affects the coupling constant $g$ and the field
$\phi$.

The effect of the action on the Lagrangian defining the theory is
calculated by writing the regularized Lagrangian in terms of
renormalized and counterterm components. The bare, or unrenormalized,
Lagrangian is \bas \mathcal{L}_B = \frac{1}{2}(|d\phi_B|^2
-m_B^2\phi_B^2) + g_B\phi_B^3\; .\eas A renormalized theory gives
Green's functions of a renormalized field, $ \phi_B = \sqrt{Z(g_B, m_B,
  z)}\phi_r$, where $\lim_{z \rightarrow 0} Z-1 = \infty$.  Then the
bare Lagrangian can be written \ba \mathcal{L}_B &=
\frac{1}{2}Z|d\phi_r|^2 -m_r^2Z\phi_r^2) + g_rZ^{3/2}\phi_r^3 \label{bare}\\ &=
\frac{1}{2}(|d\phi_r|^2 -m_r^2\phi_r^2) + g_r\phi_r^3 \nonumber \\ & \quad
\quad + \frac{1}{2}((Z - 1)|d\phi_r|^2 - (Z
-1)m_r^2\phi_r^2)+(Z^{3/2}-1)g_r\phi_r^3 \; . \nonumber \ea The second line is
called the renormalized Lagrangian, consisting of finite quantities
$\mathcal{L}_{fp}$, and the last line is the counterterm
$\mathcal{L}_{ct}$. Writing the Lagrangian as the sum \bas
\mathcal{L}_B = \mathcal{L}_{ct} + \mathcal{L}_{fp} \eas shows the
components that lead to counterterm and finite parts of the Feynman
integrals. For more details on this process see \cite{Ti}, chapters 21
and 10.

\begin{remark}
Notice that the expression for \eqref{bare} is only possible if $dZ
= 0$. This condition is known as locality of counterterms. A theory
has local counterterms if the counterterms can be expressed as
polynomials in the external momenta and mass of degree at most
$\omega(\Gamma)$. In the case of dimensional regularization, where the
counterterms are defined by the projection onto the singular part of
the Laurent series, this condition is equivalent to saying that the
counterterms are free of the energy scale. Momentum cut-off
regularization, with counterterms defined by Taylor jets taken at $0$
external momenta, as in BPHZ, also has local counterterms
\cite{IZ}\S 8.1. \end{remark}

The quantities $\phi_r$, $m_r$ and $g_r$ depend on the scale of the
theory. The differential equation \bas \beta(g_r) = \frac{1}{t}
\frac{\partial g_r}{\partial t}\eas gives the dependence of the
coupling constant on the scale. The $\beta$ function is useful in
solving the other dependencies. For a perturbative theory, it is
approximated by an asymptotic expansion by loop number of the theory. For
example, to one loop order, the perturbative $\beta$ function for a
scalar theory is \bas \beta(g) = \frac{3g^2}{16\pi^2} \; .\eas For
QED, it is \bas \beta(e) = \frac{e^3}{12\pi^2} .\eas

The perturbative $\beta$ function for a theory can change depending on
which regularization method is employed. For a scalar field theory, as
in the example computed above, and for gauge theories, such as QED,
the perturbative $\beta$ function for dimensional regularization and cut off
regularization are the same up to 3 loop orders \cite{IZ}. In this
paper, I study two perturbative $\beta$ functions: $\beta_{dr}$, from
dimensional regularization, and $\beta_{mc}$ from momentum cut-off
regularization. In section 2 of this paper, I show how these two
quantities are vector fields on a bundle generating the
renormalization group action of the respective regularization scheme.

\subsection{BPHZ on momentum cut-off regularization \label{BPHZcutoff}}

The divergences in momentum cut-off regularization can be subtracted
off by BPHZ renormalization \cite{IZ}. The BPHZ algorithm calculates
the counterterms, or divergent quantity, associated to $\Gamma$ as the
Taylor series of $\varphi(\Lambda)(e_i)(\Gamma)$ in its external
momenta, evaluated at $0$ external momenta, calculated up to
$\omega(\Gamma)$. Let $e_I$ be the multi-index $(e_{i_1}, \ldots,
e_{i_I})$, with $e_{i_j} \in \{e_1, \ldots e_{J-1}\}$, and $D_{e_I} =
\frac{\partial}{\partial e_{i_1}} \ldots \frac{\partial}{\partial
  e_{i_I}}$. Then \bas T(\varphi_{mc}(tm, te_j, t\Lambda))(\Gamma) =
\sum_{k=0}^{\omega(\Gamma)}\sum_{|I|= k}\left(D_{e_I} \varphi_{mc}(tm,
te_j, t\Lambda)(\Gamma)|_{\vec{e} = 0}\right)\frac{e_{i_1}\ldots
  e_{i_I}}{i!} \; .\eas

\begin{definition}
Let $T$ be the Taylor series operator described above, \bas T(f(e_j,
\Lambda) (\Gamma)) = \sum_{k=0}^{\omega(\Gamma)}\sum_{|I|=
  k}\frac{e_{i_1}\ldots e_{i_I}}{i!}D_{e_I} f(e_j, \Lambda)(\Gamma)
\;. \eas \end{definition}

Write the unrenormalized quantity $U(\Gamma) = \varphi_{mc}(\Lambda,
m, e_i) (\Gamma)$.  The counterterm is then \bas C_T(\Gamma) =
-T\left(U(\Gamma) + \mathop{\sum_{\gamma \subset
    \Gamma}}_{\text{divergent}} C_T(\gamma)U(\Gamma/\gamma) \right),
\eas where the sum is over all sub-graphs of $\Gamma$ that are
divergent. The graph $\Gamma/\gamma$ is obtained by replacing each
connected component of $\gamma$ by a single vertex. The renormalized
part is the sum, \bas R_T(\Gamma) = U(\Gamma) + C_T(\Gamma) +
\mathop{\sum_{\gamma \subset \Gamma}}_{\text{divergent}}
C_T(\gamma)U(\Gamma/\gamma)\;. \eas If I define a preparation map \bas
\bar{R}_T(\Gamma) = U(\Gamma) + \mathop{\sum_{\gamma \subset
    \Gamma}}_{\text{divergent}}C_T(\gamma)U(\Gamma/\gamma) \;, \eas
then $R_T= (1-T) \bar{R}_T$ and $C_T = -T \bar{R}_T$.  The
counterterms thus derived are local. Let $K_{C_T}(m, e_i,
p_j)(\Gamma)$ and $K_{\phi_{mc}}(m, e_i, p_j)(\Gamma)$ be the kernels of
the integrals $C_T(\Gamma)$ and $U(\Gamma)$. By $ C_T(\gamma)
U(\Gamma/\gamma)$, I mean the convolution product \bas C_T(\gamma)
U(\Gamma/\gamma)= \int_0^\Lambda K_{C_T}(m, e_i,
p_j)(\gamma) K_{\phi_{mc}}(m, e_i, p_j)(\Gamma/\gamma)\prod_{j=1}^l dp_j \;.\eas

Notice that if $\gamma$ has 2 external edges, then
$\omega(\Gamma/\gamma) + 2 = \omega(\Gamma)$. That is, there is an
extra vertex in the contracted graph that represents the insertion
point of $\gamma$ \cite{CK01}.

\begin{lemma}
The superficial degree of divergence of a graph is conserved under
addition of the superficial degrees of divergence of subgraphs and
contracted graphs: \bas \omega(\Gamma) = \omega(\gamma) +
\omega(\Gamma/\gamma)\;.\eas \label{degree}
\end{lemma}

\begin{proof}
Write \bas\omega(\Gamma) = dL(\Gamma) - 2I(\Gamma) = dL(\gamma) -
2I(\gamma) + dL(\Gamma/\gamma) - 2I(\Gamma/\gamma)\;.\eas This proves
the theorem.
\end{proof}

For momentum cut-off regularization, as with dimensional
regularization, the Taylor subtraction operator $T$ can be replaced by
a minimal subtraction operator $\pi$.

\begin{definition}
The minimal subtraction operator for a regularization method $\pi$ is
a projection onto only the terms in a series that are singular at a
predefined limit.
\end{definition}

For any function of this form, \bas f(\Lambda, e_i, m) = \sum_{k=
  -\infty}^n\sum_{j = 0}^\infty a_{k,j}(m, e_i) \Lambda^k
\log^j(\Lambda/m) \;,\eas the minimal subtraction operator projects
onto the terms that are ill defined as $\Lambda \rightarrow \infty$,
\bas \pi (f) = \sum_{j > 0, k\geq 0} a_{k,j}(m, e_i) \Lambda^k
\log^j(\Lambda/m) + \sum_{k > 0} a_{k,0}(m, e_i) \Lambda^k \; .\eas In
the case of momentum cut-off regularization, one can write
\cite{Collinsbook} \S5.11 \bas \varphi_{mc}(\Lambda) = \sum_{k=
  -\infty}^n\sum_{j = 0}^\infty a_{k,j}(m, e_i) \Lambda^k
\log^j(\Lambda/m) \;.\eas The preparation map associated to the
minimal subtraction is \bas \bar{R}_\pi(\Gamma) = U(\Gamma) +
\mathop{\sum_{\gamma \subset \Gamma}}_{\text{divergent}} C_\pi(\gamma)
U(\Gamma/\gamma))\;, \eas with \bas C_\pi(\Gamma) =
-\pi\bar{R}_\pi(\Gamma) \quad \textrm \quad R_\pi(\Gamma) =
(1-\pi)\bar{R}_\pi(\Gamma) \; .\eas It is a well established fact in
the physics literature that $R_\pi(\Gamma)$ is finite and that
$C_\pi(\Gamma)$ is local, for instance, see \cite{Collinsbook}. In
this paper, I present a different proof of these facts. First, I need
the following lemma.

\begin{lemma}
The quantity $(T-\pi)\bar{R}_T(\Gamma)$ is finite as $\Lambda
\rightarrow \infty$ and a polynomial in external momenta of $\Gamma$
of degree $\omega(\Gamma)$. \label{otherway}
\end{lemma}

\begin{proof}
The finite limit comes from the fact that both
$(\id-T)(\bar{R}_T(\Gamma))$ and $(\id-\pi)(\bar{R}_T(\Gamma))$ are
finite by construction. The singular terms are exactly \ba
\pi(R_T(\Gamma)) = \sum_{i>0, j>0}a_{ij}(p)\Lambda^i\log^j(\Lambda/m)
+ \sum_{ j>0}a_{0j}(p)\log^j(\Lambda/m) \label{piRt}\ea where
$a_{ij}(p)$ are polynomials of degree at most $\omega(\Gamma)$ in the
external momenta of $\Gamma$. Since $C_T(\Gamma)$ is local, this shows
that $(T-\pi)\bar{R}_T(\Gamma)$ is polynomial of at most degree
$\omega(\Gamma)$ in the external momenta.
\end{proof}

\begin{theorem}
The BPHZ preparation map, $\bar{R}_\pi$ on momentum cut-off
regularization, defines local counterterms, $C_\pi$ and finite
renormalized quantities, $R_\pi$.
\end{theorem}

\begin{proof}

The renormalized quantity is $ R_\pi(\Gamma) = (\id - \pi)
(\bar{R}_\pi(\Gamma) )$  is finite since the
operator $(\id - \pi)$ projects onto the terms that are finite as
$\Lambda \rightarrow \infty$.

The counterterm $-\pi \bar{R}_\pi(\Gamma)$ is local if
$(T-\pi)\bar{R}_\pi(\Gamma)$ is a polynomial in the external momenta
of $\Gamma$ of at most degree $\omega(\Gamma)$ and finite as $\Lambda
\rightarrow \infty$ \cite{Chetyrkyn}. I prove this by induction on the
number of subgraphs of $\Gamma$.

Define \bas f(\Gamma) = C_\pi(\Gamma) - C_T(\Gamma) \;.\eas Since
$(\id - \pi) \bar{R}_\pi(\Gamma)$ and $(\id - T) \bar{R}_T(\Gamma)$
are finite by definition, $\lim_{\Lambda \rightarrow \infty}
f(\Gamma)$ is always finite. If $\gamma$ is a graph with no
subdivergences, then \bas \bar{R}_\pi(\gamma) = \bar{R}_T(\gamma) =
U(\gamma) \;. \eas The counterterm $C_\pi(\gamma) = -\pi(U(\gamma))$
is a polynomial of homogeneous degree at most $\omega(\gamma)$ in $m$
and external momenta \cite{IZ} pg. 385.  Then $- f(\gamma) =
(T-\pi)U(\Gamma)$ is finite and of the correct degree by Lemma
\ref{otherway}.

If $\Gamma$ has a single subdivergence, $\gamma$, \ba
(T-\pi)\bar{R}_\pi(\Gamma) = (T-\pi)U(\Gamma) + (T-\pi) C_\pi(\gamma)
U(\Gamma/\gamma)\; . \label{1sd} \ea The graph $\Gamma/\gamma$ has no
subdivergences. This can be rewritten \bas(T-\pi)\bar{R}_\pi(\Gamma) = (T-\pi)\bar{R}_T(\Gamma) +
(T-\pi) f(\gamma) U(\Gamma/\gamma) \eas Writing \bas U(\Gamma/\gamma) = R_T(\Gamma/\gamma) + f(\Gamma/\gamma) - C_\pi(\Gamma/\gamma) \eas gives \bas (T-\pi)\bar{R}_\pi(\Gamma) = (T-\pi)\bar{R}_T(\Gamma) -
 (T-\pi)\left[f(\gamma)(C_\pi(\Gamma/\gamma) -f(\Gamma/\gamma) -
  R_T(\Gamma/\gamma))\right] \; . \eas From Lemma \ref{otherway},
$(T-\pi)\bar{R}_T(\Gamma)$ is finite and of the correct degree. By
definition of the operators $\pi$ and $T$, $\pi(f(\gamma)
(f(\Gamma/\gamma)F R_T(\Gamma/\gamma)))= 0$ since both $f(\gamma)$ and
$f(\Gamma/\gamma)+ R_T(\Gamma/\gamma)$ are finite. The term
$T(f(\gamma)(f(\Gamma/\gamma) - R_T(\Gamma/\gamma))$ is a polynomial
of the correct degree (by virtue of the external $T$) and finite,
since all the components are finite. It remains to examine
$(T-\pi)(f(\gamma)C_\pi(\Gamma/\gamma))$. The term $f(\gamma)$ is a
polynomial of degree at most $\omega(\gamma)$ and $C_\pi(\Gamma/\gamma)$ is a
polynomial of degree at most $\omega(\Gamma/\gamma)$. Therefore,
$f(\gamma)C_\pi(\Gamma/\gamma)$ is a polynomial of degree at most
$\omega(\Gamma/\gamma) + \omega(\gamma) = \omega(\Gamma)$, and \bas
(T-\pi) f(\gamma)C_\pi(\Gamma/\gamma) = (\id -\pi)
f(\gamma)C_\pi(\Gamma/\gamma) \eas which is finite by definition of
$\pi$ and a polynomial of the correct degree. Thus \bas (T-\pi) \bar{R}_\pi(\Gamma) \eas  is a
polynomial of degree $\omega(\Gamma)$, and finite as $\Lambda
\rightarrow \infty$.

This implies that $f(\Gamma)$ is finite and of degree at most
$\omega(\Gamma)$ for $\Gamma$ a graph with $1$ divergent
subdiagram. Write \bas f(\Gamma) = (T-\pi)\bar{R}_\pi(\Gamma) +
(T-\pi)\bar{R}_T(\Gamma) - T(\bar{R}_\pi(\Gamma)) +
\pi(\bar{R}_T(\Gamma) \; .\eas the first term is finite and of the
correct degree by the argument above. The second term is finite and of
the correct degree by Lemma \ref{otherway}. The third term is of the
correct degree by definition of the operator $T$. By the explicit
expression in \eqref{piRt}, the last term is also of the correct
degree.

Suppose $f(\Gamma)$ is a polynomial of degree at most $\omega(\Gamma)$
in the external momenta, for all $\Gamma$ with fewer than $n$
subdivergences. Then,
\bas(T-\pi)\bar{R}_\pi(\Gamma) =  (T-\pi)\left[
  \bar{R}_T(\Gamma) + \sum_{\gamma \subset
    \Gamma}f(\gamma)U(\Gamma/\gamma) \right]\; ,\eas where each divergent subgraph $\gamma$
has $j<n$ subdivergent graphs. This is finite and of the correct degree
by the arguments presented in the $n = 1$ case. Furthermore,
$f(\Gamma)$ is also finite and a polynomial of the correct degree
for $\Gamma$ a graph with $n$ divergent subgraphs.

\end{proof}

Connes and Kreimer use BPHZ renormalization on dimensional
regularization with the minimal subtraction operator instead of the
Taylor series operator in their work \cite{CK00, CK01}. In this paper,
I extend their work to include cutoff regularization. The substitution
of the minimal subtraction operator for the Taylor series operator in
BPHZ renormalization of different regularization schemes is well
established. For example, Collins \cite{Collins75} does so for
dimensional regularization, Speer \cite{Speer74} for analytic
regularization, and \cite{Collinsbook} \S 5.11.3 or \cite{Reisz88} for
lattice regularization. In the dimensional regularization the
regularized integral is of the form \bas \varphi_{dr}(z)(\Gamma) =
\sum_{i= -n}^\infty a_iz^i \;.\eas The minimal subtraction operator
projects onto the polar part of the Laurent series \bas \pi \circ
\varphi_{dr}(z)(\Gamma) = \sum_{i= -n}^{-1} a_iz^i \;.\eas

Having established that both dimensional regularization and momentum
cut-off regularization can be renormalized by BPHZ renormalization
under a minimal subtraction operator, I define the target algebras of
the regulated integrals such that the same minimal subtraction
operator suffices for both regulation schemes. The regulator for
dimensional regularization, $z$, is a complex parameter, where as the
regulator for momentum cut-off, $\Lambda$ is real. First I show that
the regulator for momentum cut-off can be extended to a complex
parameter as well.

\begin{theorem} 
As a complex regulation scheme, cut-off momentum regulation is
identical to the real case. For $\Lambda \in \R$, \bas
\varphi_{mc}(\Lambda e^{i\theta})(\Gamma) = \int_C\prod_{k=1}^I
\frac{1}{f_k(p_i,e_j)\cdot \overline{f_k(p_i,e_j)} + m^2}
\prod_{i=1}^l d^4p_i = \varphi_{mc}(\Lambda)(\Gamma)\;. \eas \end{theorem}

\begin{proof}
To analytically continue the momentum cut-off regulator, consider the
Feynman integral with complex momentum \bas
\int_{\R^{4l}}\prod_{k=1}^I \frac{1}{f_k(p_i,e_j)\cdot
  \overline{f_k(p_i,e_j)} + m^2} \prod_{i=1}^ld^4p_i \; .\eas Then for
a complex cutoff regulator $\Lambda e^{i\theta} \in \C$, the regulated
Feynman integral is \bas \varphi_{mc}(\Lambda e^{i\theta})(\Gamma) =
\int_C\prod_{k=1}^I \frac{1}{f_k(p_i,e_j)\cdot \overline{f_k(p_i,e_j)}
  + m^2} \prod_{i=1}^l d^4p_i \eas taken along the contour $C = t
e^{i\theta}$ for $t \in [-\Lambda, \Lambda]$. The symmetries of the
integrand give $\varphi_{mc}(\Lambda)(\Gamma) = \varphi_{mc}(\Lambda
e^{i\theta})(\Gamma)$.
\end{proof}

Notice that \ba \varphi_{dr}(z)(\Gamma) \in
\C[z^{-1}][[z]] \label{dralg}\ea for any Feynman diagram
$\Gamma$. Rewrite $\Lambda e^{i\theta} = 1/z$. Then
\ba\varphi_{mc}(z)(\Gamma) \in \C[z^{-1}, \log
  (m|z|)][[z]]\label{mcalg}\ea for any Feynman diagram $\Gamma$. I
keep the factor of $m$ in the logarithmic term for dimensional
considerations. Next I define an algebra $\scrA$ containing both the
algebras defined in \eqref{dralg} and \eqref{mcalg}.

\begin{definition}
Define $ \scrA := \C[z^{-1}][[z]][[y]]/(e^y -zm)$ to be the target
algebra for dimensional regularization and momentum cut-off
regularization.
\end{definition}

For any Feynman diagram $\Gamma$, $\varphi_{dr}(\Gamma),
\varphi_{mc}(\Gamma) \in \scrA$. Any element $f \in \scrA$ can be
written $f = \sum_{j=0}^\infty\sum_{i= -n}^{\infty} a_{ij}
z^iy^j$. The minimal subtraction operator is a projection onto the
subalgebra of $\scrA$ that contains only the term that are singular at
$z = 0$.

\begin{definition}
Let $\pi$ be the the projection on $\scrA$ \bas \pi : \scrA
&\rightarrow& \scrA_-: = (z^{-1}\C[z^{-1}][[y]]\oplus y\C[[y]])/(e^y -
zm) \\ \sum_{j=0}^\infty\sum_{i= -n}^{\infty} a_{ij} z^iy^j &\mapsto&
\sum_{i<0} a_{ij} z^iy^j + \sum_{j >0} a_{0j} y^j \eas that maps to
the subalgebra of $\scrA$ that is singular at $z= 0$, (as $y
\rightarrow -\infty$). Define $\scrA_+$ to be the subalgebra such that
$\scrA = \scrA_- \oplus \scrA_+$. This is the subalgebra of terms that
are finite at $z=0$.
\end{definition}

This is the same as the minimal subtraction operator for momentum
cutoff regularization.

\begin{theorem}
The operator $\pi: \scrA \rightarrow \scrA_-$ restricts to minimal
subtraction operator for dimensional regularization.
\end{theorem}
\begin{proof}
Define $A =\C[z^{-1}][[z]] \subset \scrA$. This is the target algebra
of dimensional regularization. Since this the exactly the space where
the powers of $y$ are $0$, restricting $\pi$ to this domain gives \bas
\pi|_A : A \rightarrow A_- := \C[z^{-1}] \;.\eas Write $A_+ = \C[[z]]
\subset \scrA_+$. Thus the projection map $\pi$ restricts to the
minimal subtraction operator for dimensional regularization.
\end{proof}

In the rest of the paper, I apply the methods of \cite{CK00},
\cite{CK01}, and \cite{CM06} to build a Hopf algebra of Feynman
diagrams, define the counterterms using Birkhoff decomposition, and
define the $\beta$ function for cut-off regularization on the
corresponding renormalization bundle.

\section{The Renormalization Bundle }
In \cite{CK00}, Connes and Kreimer build a Hopf algebra, $\h$, out of
the divergence structure of the Feynman diagrams for a scalar field
theory under dimensional regularization. They use the BPHZ algorithm
to renormalize the theory, replacing Taylor subtraction around $0$
external momenta with the minimal subtraction operator. The key to
constructing this Hopf algebra is the sub-divergence structure of the
graphs as defined by power counting arguments. The co-product of the
Hopf algebra is defined to express the same sub-divergence data as in
Zimmermann's subtraction formula for BPHZ renormalization
\cite{IZ}. Replacing the Taylor series operator in BPHZ for the
minimal subtraction operator does not change the divergence structure
of the diagrams. Therefore, I use the same Hopf algebra to study
cut-off regularization. In \cite{vS06}, van Suijlekom constructs a
Hopf algebra that captures the renormalization structure of QED under
dimensional regularization. This is the same Hopf algebra that is
needed to study QED under cut-off regularization. The arguments in
this paper apply to both scalar $\phi^4$ and QED, even though cut-off
regularization does not preserve the gauge symmetries of QED.

To briefly recall notation, let \bas \h = \C[\{1PI \text{ graphs with
    2 or 4 external edges}\}] \eas  be the Hopf algebra of Feynman
diagrams, with multiplication defined by disjoint union. It is graded
by loop number, with $Y$ the grading operator. If $\Gamma \in \h_n$,
$Y(\Gamma) = n\Gamma$. The co-unit $z$ is $0$ on $\h_{\geq 1}$, and is
the identity map on $\h_0$. An admissible sub-graph of a 1PI graph,
$\Gamma$ is a graph, $\gamma$, or product of graphs, that can be
embedded into $\Gamma$ such that each connected component has 2 or 4
external edges. The graph $\Gamma/\gamma$ is the graph obtained by
replacing each connected component of $\gamma$ with a vertex. The
admissible sub-graphs correspond to the divergences subtracted by
Zimmermann's subtraction algorithm. Using Sweedler notation, the
co-product on $\h$ is given by the sum \bas \Delta \Gamma = 1 \otimes
\Gamma + \Gamma \otimes 1 +\sum_{\gamma \text{ admis}}\gamma \otimes
\Gamma/\gamma \; .\eas Let $\epsilon$ and $\eta$ denote the co-unit
and unit of this Hopf algebra.

The Hopf algebra is connected and each graded component $\h_n$ is
finitely generated as an algebra. Write the graded dual of this Hopf
algebra $\h^* = \oplus_n \h_n^*$. The product on $\h^*$ is the
convolution product $f\star g (\Gamma) = m(f \otimes g)
\Delta(\Gamma)$. The antipode, $S$, on the restricted dual defines the
inverse of a map under this convolution product, $f^{\star -1} = S
(f)$. By the Milnor-Moore theorem, $\h^* \simeq \mathcal{U}(\g)$ is
isomorphic to the universal enveloping algebra of the Lie algebra
$\g$, generated by the infinitesimal derivatives \bas
\delta_\Gamma(\Gamma') = \begin{cases} 1 & \Gamma = \Gamma' \text{
    1PI} \\ 0 & \Gamma \neq \Gamma' \end{cases} \;.\eas The generators
of the Lie algebra are infinitesimal characters \bas
\delta_\Gamma(\gamma\Gamma') = \epsilon(\gamma) \delta_\Gamma(\Gamma')
+ \epsilon(\Gamma') \delta_\Gamma(\gamma) \;.\eas The Lie bracket is
given by $[f, g] = f\star g - g\star f$. The corresponding Lie group
$G = e + \g$ is the group of algebra homomorphisms $\Hom_{alg}(\h, \C)
= \Spec \h$. See \cite{EM} for more discussion of this Lie group and
Lie algebra.

In this paper, I study regularization procedures that induce maps from
 the Hopf algebra $\h$ to the algebra generated by the regulation
parameter, $\scrA$.  In dimensional regularization, $z$ corresponds to
the complex ``dimension'' regulator. In momentum cut-off regulation,
$z$ corresponds to the complexification of the inverse of the cut-off,
$z= \frac{e^{-i\theta}}{\Lambda}$, and polynomials in $y$ correspond
to polynomials in $\log (|z|m)$.  Minimal subtraction on both these
regulation schemes is encoded by considering the direct sum
decomposition $\scrA = \scrA_- \oplus \scrA_+$, where $\scrA_-=
(z^{-1}\C[z^{-1}][[y]]\oplus y\C[[y]])/(e^y - zm)$. The projection map
\bas \pi : \scrA \rightarrow \scrA_-\eas is the subtraction map used
in minimal subtraction for both schemes. This projection map is a
Rota-Baxter operator on $\scrA$. The algebra $\scrA$ and this
Rota-Baxter operator are discussed in detail in \cite{ManchonPaycha}
in the context of cut-off regularization, and other applications.

\begin{definition} A
  Rota-Baxter operator, $R$, of weight $\theta$ on an algebra $A$ is a
  linear map \bas R: A \rightarrow A \eas that satisfies the
  relationship \bas R(x)R(y) + \theta R(xy) = R(xR(y)) + R(R(x)y)
  \;.\eas The pair $(A, R)$ is called a Rota-Baxter
  algebra. \end{definition}

\subsection{Generalization of Birkhoff decomposition}

Let $\varphi_{mc} , \varphi_{dr} \in \Hom_{alg}(\h, \scrA)$ be the algebra
homomorphisms from $\h$, the Hopf algebra of Feynman graphs, to
$\scrA$ the algebra spanned by the regulating parameters corresponding
to momentum cut-off regularization and dimensional regularization
respectively. Paralleling the work of Connes and Kreimer in
\cite{CK00}, I write the counterterm and the renormalized part
of cut-off regularization and dimensional regularization under minimal
subtraction as a Birkhoff-type decomposition of $\varphi_{mc}$ and
$\varphi_{dr}$.

Ebrahimi-Fard, Guo and Kreimer show
that, if the algebra $\scrA$ is endowed with a Rota-Baxter
operator, $R$, there is an unique expression for each $\varphi \in
\Hom_{alg}(\h, \scrA)$ as $\varphi_- \star \varphi_+$ such that
$\varphi_-$ lies in the image of $R$, if $x \in \h$, and $\varphi_-$,
$\varphi_+ \in \Hom_{alg}(\h, \scrA)$. If $R$ corresponds to a
subtraction operator for BPHZ, $\varphi_-(x)$ corresponds to the
counterterm of $x$ and $\varphi_+(x)$ the renormalized part
\cite{EGK04}. The following theorem follows directly from this result.

\begin{theorem}
Let $\varphi \in \Hom_{alg}(\h,\scrA)$. Define the projection map
$\pi: \scrA \rightarrow \scrA_-$. There is a unique decomposition of
each $\varphi = \varphi_-^{\star -1} \star \varphi_+$ with
$\varphi_-(\Gamma) \in \scrA_-$ for $\Gamma \in \ker \epsilon$,
$\phi_-(1)= 1$ and $\varphi(\Gamma) \in G(\scrA_+)$.  \end{theorem}

\begin{proof} Notice that $\pi:\scrA \rightarrow \scrA_-$ is a
Rota-Baxter operator of weight 1. Let $\Hom(\h, \scrA)$ be the algebra
of linear maps from $\h$ to $\scrA$, with point-wise multiplication and
unit $e = \eta_\scrA \circ \epsilon$. For $\varphi \in \Hom(\h, \scrA)$,
let $R = \pi \circ \varphi$. Then $R$ is a Rota-Baxter operator on
$\Hom(\h, \scrA)$. By extending the convolution product on $\h^*$ to
$\Hom(\h, \scrA)$, each algebra homomorphism $\varphi \in \Hom_{alg}(\h,
\scrA)$ can be uniquely decomposed according to $\pi$. For all $\Gamma
\in \ker(\epsilon)$, \bas \varphi_-(\Gamma) = -\pi(\varphi(\Gamma) +
\sum_{\gamma \text{ admis.}}\varphi_-(\gamma)\varphi(\Gamma//\gamma)
\\ \varphi_+(\Gamma) = (e-\pi)(\varphi(\Gamma) + \sum_{\gamma \text{
    admis.}}\varphi_-(\gamma)\varphi(\Gamma//\gamma)) \;. \eas The maps
$\varphi$, $\varphi_-$ and $\varphi_+$ are algebra homomorphisms from $\h$ to
$\scrA$, $\C \oplus \scrA_-$ and $\scrA_+$ respectively. That is, \bas
\varphi(1) = \varphi_-(1) = \varphi_+(1) = 1_\scrA\; . \eas However, for
$\Gamma \in \ker(\epsilon)$, $\varphi_-(\Gamma) \in \scrA_-$. \end{proof}

This is a generalization of the Birkhoff decomposition theorem, which
says that for any simple closed curve, $C$, in $\CP^1$ that does not pass
through $0$ or $\infty$, and a map \bas \varphi : C \rightarrow G \;,\eas
for a complex Lie group $G$, there is a function $\varphi_-$ that is
holomorphic on the connected component of $\CP^1 \setminus C$ that
contains $\infty$ and a function $\varphi_+$ that is holomorphic on the
connected component of $\CP^1 \setminus C$ that contains $0$, such
that $\varphi= \varphi_- \star \varphi_+$. In the setting of dimensional
regularization, $\varphi_{dr} \in \Hom_{alg}(\h, \C[z^{-1}][[z]]) =
G(\C[z^{-1}][[z]])$, is viewed as a map from a loop in $\Spec
\C[z^{-1}][[z]] \subset \C$ to $G = \Spec \h$. The Birkhoff
decomposition theorem on loops directly gives the existence of such a
decomposition. The Rota-Baxter algebra argument in
\cite{EGK04} generalizes the Birkhoff decomposition setting to other
algebras.

\subsection{The geometric $\beta$ function}
So far, I have considered $\varphi_{mc}$ and $\varphi_{dr} \in
\Hom_{alg}(\h, \scrA)$ to be sections of a (trivial) $G$ principal
bundle over $\Spec \scrA$. Call this bundle $K \simeq G\times \Spec
\scrA \rightarrow \Spec \scrA$. Sections of this bundle correspond to
algebra homomorphisms from $\h$ to $\scrA$. 

I follow the notation in \cite{CM06} and consider a complex
renormalization group, instead of a real group. For the following
arguments the renormalization group is $\C^\times$.

\begin{definition}
Define the renormalization group as $\C^\times$. Parametrize
it by $t = e^s$ for $s\in \C$.  \end{definition}

Let $\sigma$ be a one parameter family of diffeomorphism on $G(\scrA)$
written \bas \sigma: \C^\times \times G(\scrA) &\rightarrow& G(\scrA)
\\ (t, \varphi(z, y)) &\mapsto& \sigma_t\varphi(z,y) \;. \eas These
one parameter family of diffeomorphisms are a natural generalization
of the renormalization group action on a regularized QFT, written as
an element of $G(\scrA)$. 
\begin{definition} Write $\sigma_t \in
\textrm{Diff}(G(\scrA))$. The orbit under this group of
diffeomorphisms, \bas \{\sigma_t(\varphi(z,y)) | t \in \C^\times\}\;,
\eas of $\varphi(z,y) \in G $ is a one parameter curve in $G(\scrA)$ which I also denote $ \sigma_t (\varphi(z,y))$ in $G$. \end{definition}

Consider the bundle \bas P \simeq K \times \C^\times \rightarrow B
\simeq \Spec \scrA \times \C^\times\; .\eas Sections of this bundle
are algebra homomorphism from $\h$ to $\scrA \otimes \C[t^{-1}, t]$. I
write these sections $\psi: B \rightarrow P$ as \bas \psi (z, y, t)
\in \Hom_{alg}(\h, \scrA[t, t^{-1}]) = G(\scrA[t^{-1},t])\; .\eas The
renormalization group acts, $\C^\times$ acts on $G(\scrA[t^{-1}, t])$,
the sections of $P\rightarrow B$, as \bas \C^\times \times
G(\scrA[t^{-1}, t]) &\rightarrow G(\scrA[t^{-1}, t]) \\ (t,
\psi(z,y, u)) &\mapsto \psi(z,y,tu) \;. \eas The bundle
$P\rightarrow B$ is not $\C^\times$ equivariant. Each one parameter diffeomorphism of $G(\scrA)$ defines a $\C^\times$ subbundle of $P$.

\begin{theorem}
For every one parameter diffeomorphism $\sigma$, there is a
$\C^\times$ equivariant $G$ principal bundle $K_\sigma \rightarrow B$
with sections corresponding to the curves $\sigma_t(\phi(z,y))$ for
$\varphi(z,y) \in G(\scrA)$.
\end{theorem}

\begin{proof} 
Let $K_\sigma$ be the trivial $\C^\times$ bundle over $K$ defined by
the one parameter diffeomorphism $\sigma$. This is a $\C^\times$
equivariant bundle over $K$. For $\varphi(z, y) \in G(\scrA)$,
$(u, z, y, \varphi(z, y)) \in K_\sigma$, and $t \in \C^\times$, \bas
\sigma : \C^\times \times K_\sigma &\rightarrow K_\sigma \\ (t, (u, z,
y, \varphi(z,y))) &\mapsto (tu, z, y, \sigma_t(\varphi(z,y))) \;.\eas
The one parameter diffeomorphism $\sigma$ defines a 
curve $\varphi_\sigma(z,y,t)$ in $K_\sigma \rightarrow K$. 

I can instead view $K_\sigma$ as a $G$ principal bundle over $B$. The
sections of this bundle are of the form \bas (t, \varphi(z, y)) :
B\rightarrow K_\sigma \;. \eas These sections are compatible with the
renormalization group action defined by $\sigma$.  For $u \in
\C^\times$, \bas \sigma_u(t, \varphi(z,y)) = (ut,
\sigma_u(\varphi(z,y))).\eas Therefore the bundle $K_\sigma
\rightarrow B$ is $\C^\times$ equivariant.  This is the desired
construction.
\end{proof} 

The sections of $K_\sigma \rightarrow B$ form the group
$\tilde{G}_\sigma(\scrA) := \C^\times \rtimes_\sigma G(\scrA)$ defined
by the semi-direct product of $G(\scrA)$ with $\C^\times$ under the
action $\sigma$. There is a bundle injection \bas \xymatrix{ K_\sigma
  \ar@{^{(}->}[rr]^{i_\sigma} \ar[dr]_\pi & &P \ar[dl]^\pi \\ & B & }
\eas defined on sections of the bundles as \bas i_\sigma: \C^\times
\times G(\scrA) &\rightarrow G(\scrA[t^{-1}, t]) \\ (t, \varphi(z, y))
& \mapsto \sigma_t(\varphi(z, y)) \; .\eas The injection fixes a unit
for the action of the renormalization group, mapping \bas i_\sigma(1,
\varphi(z,y)) = \sigma_1\varphi(z, y, 1): = \psi(z,y,1) \;.\eas In this manner
\emph{all} possible one parameter diffeomorphisms on
$G(\scrA)$ can be represented as sections of $P \rightarrow B$.


\begin{theorem}
Let $\sigma$ be a one parameter diffeomorphism on $G(\scrA)$ that
corresponds to a renormalization group action on a QFT. The geometric
$\beta$ function of $\sigma$ is the vector field $\beta_\sigma \in
TG(\scrA)$ defined by the logarithmic differential on $K$, \bas
\beta_\sigma(\varphi(z,y)) = \left[\sigma_t(\varphi(z, y))^{-1} \star
  \dt \sigma_t(\varphi(z,y))\right]|_{t=1}\;.\eas This vector field
generates the one parameter diffeomorphism on $G(\scrA)$, $\sigma$.
 \end{theorem}

\begin{proof}
First I check that $\beta_\sigma \in TG(\scrA)$ by verifying that
$\beta_\sigma (\varphi) \in \g(\scrA)$ for all $\varphi(z,y) \in
G(\scrA)$. For $a,b \in \h$, I write $\Delta(a) = \sum_{(a)}a' \otimes
a''$.  Then \bas \beta_\sigma(\varphi)(ab) = \varphi(z, y)^{-1} \star
\dt \sigma_t\varphi(z, y)|_{t=1} (ab) = \\ \sum_{(a)(b)}\varphi(z,
y)^{-1} (a')\varphi(z, y)^{-1} (b') \dt \left(\sigma_t\varphi(z, y)
(a'')\sigma_t\varphi(z, y) (b''))\right)|_{t=1}\;. \eas by definition
of the coproduct. Calculating the derivative gives \bas
\sum_{(a)(b)}\varphi(z, y)^{-1} (a')\varphi(z, y)^{-1} (b') \times
\\ \left[ (\dt \sigma_t\varphi(z, y) (a''))|_{t=1}\varphi(z, y) (b'')
  + \varphi(z, y) (a'')\dt \sigma_t\varphi(z, y) (b''))|_{t=1} 
  \right] \eas which rearranges to \bas \sum_{(a)(b)}\varphi(z,
y)^{-1} (a')(\dt \sigma_t \varphi(z, y) (a''))|_{t=1} \varphi(z,
y)^{-1} (b') \varphi(z, y) (b'') +\\ \varphi(z, y)^{-1}
(a')\varphi(z, y) (a'') \varphi(z, y)^{-1} (b') (\dt\varphi(z,
y) (b''))|_{t=1}  \\ = \beta_\sigma(\varphi)(a) \varepsilon(b)
+\varepsilon(a) \dt \beta_\sigma (\varphi)(b)\;. \eas Thus
$\beta(\varphi) \in \g(\scrA)$. The second statement comes from
definition.
\end{proof}

It remains to check that $\beta_\sigma$ defines a one to one
correspondence between $G(\scrA)$ and $\g(\scrA)$ for each $\sigma$.

\begin{definition}
The renormalization group action of $\sigma$ on $G(\scrA)$ induces a
diffeomorphism on $\g(\scrA)$. Let $\alpha_\sigma(z,y,t)$ be a the one
parameter path through $\alpha \in \g(\scrA)$ induced by $\sigma$,
such that $\alpha_\sigma(z,y,1) = \alpha$.
\end{definition}

\begin{lemma}
The $\beta$ function on the group action $\sigma$ defines a set
bijection\bas \beta_\sigma: G(\scrA) \rightarrow
\g(\scrA)\;\;. \eas \label{betainv} \end{lemma}

This is just an extension of the proof of the similar statement in
\cite{EM, thesis1}.
\begin{proof}
We have shown that $\beta_\sigma(\varphi) \in \g(\scrA)$ for all
$\varphi \in G(\scrA)$.

To see that this is a bijection, I define an inverse function,
$\rho_\sigma$, and check that for all $\alpha \in \g(\scrA)$,
$\rho_\sigma(\alpha)$ is well defined and in $G(\scrA)$.

For any $\alpha \in \g(\scrA)$, and one parameter diffeomorphism $\sigma$, define a map $\rho_\sigma$, \bas
\rho_\sigma : \g(\scrA) \rightarrow G(\scrA) \eas recursively by
\ba \dt \left(\rho_\sigma(\alpha_\sigma(z,y,t)\right) =
\rho_\sigma(\alpha_\sigma(z,y,t)) \star
\alpha_\sigma(z,y,t) \label{recursiveinverse} \ea with the initial
condition that \bas \rho_\sigma(t)(\alpha_\sigma(z,y,t)) (1) = 1 \quad
\forall t \in \C^\times \;. \eas It is easy to check that \bas
\rho_\sigma(1)(\beta_\sigma(\varphi)) = \varphi \; ,\eas for $\varphi
\in G(\scrA)$. If it is well defined with appropriate domain, it is
the desired inverse function.

It remains to check that $\rho_\sigma$ is well defined for all
$\alpha \in \g(\scrA)$ and has a domain contained in $G(\scrA)$. This
I do recursively on the grading on $\h$. I use the notation \bas
\Delta(a) = 1\otimes a + a \otimes 1 + \sum_{(a)}a' \otimes a''
\;.\eas To ease notation, let $\psi_\sigma(z,y,t) \in\{
\rho_\sigma(\alpha_\sigma(z,y,t))\}$. If $a \in \h_0$, both sides of
\eqref{recursiveinverse} are $0$. If $a \in \h_1$,
\eqref{recursiveinverse} becomes \bas \psi_\sigma(z,y,t_0)(a) =
\int_0^{t_0}\frac{\alpha_\sigma(z,y,t)(a)}{t} dt \;. \eas Therefore,
$\psi_\sigma(z,y,t)$ is well defined on $\h_1$. By
induction, for $a \in \h_n$, \eqref{recursiveinverse} becomes \bas
\psi_\sigma(z,y,t_0)(a) = \int_0^{t_0}\frac{\alpha_\sigma(z,y,t)(a) +
  \sum_{(a)}\psi_\sigma(z,y,t)(a')\alpha_\sigma(z,y,t)(a'')}{t} dt
\;. \eas Since $Y(a'), Y(a'') < n$, the kernel of the integral is known by
induction, and $\psi_\sigma(z,y,t)$ is unique.

To see that $\psi_\sigma(z,y,t) \in G(\scrA)$, consider by induction
on loop number its action on composite elements $ab \in \h$, where
$a, \;b\in \h$  are indecomposables. For $a, b \in \h_1$, \bas \dt
\psi_\sigma(z,y,t)(ab) = \psi_\sigma(z,y,t)(b)\alpha_\sigma(z,y,t)(a)
+ \psi_\sigma(z,y,t)(a)\alpha_\sigma(z,y,t)(b)
\\ \psi_\sigma(z,y,t)(b) \dt\psi_\sigma(z,y,t)(a) +
\psi_\sigma(z,y,t)(a) \dt\psi_\sigma(z,y,t)(b) = \dt
(\psi_\sigma(z,y,t)(a) \psi_\sigma(z,y,t)(b) \;.\eas For general
indecomposables $a, \;b\in \h$, \bas \dt \psi_\sigma(z,y,t)(ab) =
\psi_\sigma(z,y,t)(b)\alpha_\sigma(z,y,t)(a) +
\psi_\sigma(z,y,t)(a)\alpha_\sigma(z,y,t)(b) +
\\ \sum_{(a)}\psi_\sigma(z,y,t)(ba')\alpha_\sigma(z,y,t)(a'') +
\sum_{(b)}\psi_\sigma(z,y,t)(ab')\alpha_\sigma(z,y,t)(b'') \\ =
\psi_\sigma(z,y,t)(b)\alpha_\sigma(z,y,t)(a) +
\psi_\sigma(z,y,t)(a)\alpha_\sigma(z,y,t)(b) +\\ \psi_\sigma(z,y,t)(b)
\sum_{(a)} \psi_\sigma(z,y,t)(a' )\alpha_\sigma(z,y,t)(a'') +
\psi_\sigma(z,y,t)(a) \sum_{(b)} \psi_\sigma(z,y,t)(b')
\alpha_\sigma(z,y,t)(b'') \;.\eas The second equality comes from the
induction step, since $Y(ab') < Y(ab)$ and $Y(ba') < Y(ab)$. Thus for
all $a, b \in \h$, \bas \psi_\sigma(z,y, t) (ab) = \psi_\sigma(z,y, t)
(a) \psi_\sigma(z,y, t) (b) \eas showing that it is in $G(\scrA)$ as
desired.

\end{proof}

The geometric $\beta$ function and the perturbative $\beta$ function
for a theory are related objects. The geometric object $\beta_\sigma$
is the generator for a specified one parameter family of
diffeomorphisms $\sigma$. The perturbative
$\beta$ function, on the other hand, calculates the scale dependence
of a regularized QFT.

Recall that regularized Lagrangians are no longer scale
invariant: \bas \int_{\R^n} \mathcal{L}(x, z, y) d^nx \neq
\int_{\R^n} \mathcal{L}(tx, z, y) d^n(tx) \; .\eas The action of the
renormalization group on a regularized Lagrangian is defined by the
scale dependence of the regularized theory. The $\beta$ function of a
QFT is measured (perturbatively) with respect to this action.

\begin{definition}
The perturbative $\beta$ function is the geometric $\beta$ function,
$\beta_\sigma(\varphi)$ defined on a pair $(\varphi, \sigma)$, where
$\varphi$ corresponds to the Feynman rules under an appropriate
regularization scheme, and the one parameter family of diffeomorphisms $\sigma$
is defined by the scale dependence of the Feynman rules introduced by
the regularization scheme. \end{definition}

I relate the perturbative $\beta$ functions for dimensional
regularization and momentum cut-off regularization under these
conventions.

The relevant renormalization group action for dimensional
regularization is \bas \sigma_{dr,t}\varphi(z, y) =
t^{zY}\varphi(z,y)\; ,\eas as defined by \eqref{scalingdr}. If $t =
e^s$, for $s \in \C$, the relevant renormalization group action for
momentum cut-off is \bas \sigma_{mc,t}\varphi(z, y) = \varphi(tz_0,
y_0 + s)\; ,\eas for some fixed $z_0$ where $e^{y_0} = z_0$ as defined
by \eqref{scalingmc}.

\begin{remark}
In this notation, the $\beta$ function for dimensional regularization
is \bas \beta_{\sigma_{dr}}(\varphi_{dr}) = \varphi_{dr}^{-1} \star
zY\varphi_{dr} \;, \eas which is actually $z \beta$, where $\beta$ is
the relevant $\beta$ function defined in \cite{CK01}. It is further
worth noting that due to the form of the renormalization group action
for dimensional regularization, $\beta_{\sigma_{dr}}$ defines a set
bijection between $G(\scrA)$ and $\g(\scrA)$ \cite{EM}. This bijection
does not exist for all $\beta_\sigma$.
\end{remark}

The fiber over every point $(z, y, \varphi) \in K$ in $K_\sigma
\rightarrow K$ represents the energy scale of the character. The
initial energy scale of which an physical theory is evaluated
corresponds to a section $\varphi(z, y, 1)$.  Two different sections
$\varphi(z, y)$ and $\eta(z,y) \in G(\scrA)$ represent the same
character at different energy scales if there exists a one parameter
family of diffeomorphisms $\sigma$ and a $t \in \C^\times$ such that
\bas \sigma_t \varphi(z, y) = \eta(z, y)\;. \eas I show this
explicitly in the case of momentum cut-off regularization and
dimensional regularization.

Fix the regulators $z_0$ and $z_o'$. For momentum cut-off
regularization, write $z(t) = z_0t$ and $ z'(t) = z_0't $, where $z_0
= u z_0'$ and $u = e^v$. Let the character $\varphi_{mc, z_0}(z(t),
y(s))$, correspond to the Feynman rules cut off at the momentum
$\frac{1}{z(t)}$ at the energy scale $\frac{1}{z_0}$. Then \bas
\sigma_{mc,u}\varphi_{mc, z'_0}(z'(t), y'(s)) =\varphi_{mc,
  z_0}(z'(tu), y'(s+v)) = \varphi_{mc, z_0}(z(t), y(s))\; .\eas

In the case of dimensional regularization, write the character
associated to the field theory evaluated at the energy scale $u$ as
$\varphi_{dr,u}(z)$. The action of the renormalization group is given
by \bas \sigma_{dr, t} \varphi_{dr,u/t}(z) = t^{zY}\varphi_{dr,u/t}(z)
= \varphi_{dr,u}(z).\eas The subtlety of keeping track of the
renormalization scale is a minor point in the case of dimensional
regularization because of the independence of the regulator and the
energy scale. The notation keeping track of the energy scale is
dropped, but it is implicit in the definition of equisingularity in
\cite{CM06}.

The perturbative $\beta$ functions for dimensional regularization and
momentum cut-off regularization are defined by the vector fields
$\beta_{\sigma_{dr}}$ and $\beta_{\sigma_{mc}}$
respectively. Evaluating them on characters $\varphi_{dr,u}(z)$ and
$\varphi_{mc, z_0}(z(t), y(s))$ respectively gives the corresponding perturbative $\beta$
function.  For dimensional regularization, \bas
\beta_{\sigma_{dr}}(\varphi_{dr,u}(z)) = \varphi^{-1}_{dr,u}(z) \star
\dt t^{zY}\varphi_{dr,u}(z) |_{t=1} = z\varphi^{-1}_{dr,u}(z) \star Y
\varphi_{dr,u}(z)\;. \eas This vector field is a function of the
regulator $z$, the complex dimension. That is $
\beta_{\sigma_{dr}}(\varphi(z)) \in \g(\scrA)$ for all $\varphi \in
G(\scrA)$. The physical $\beta$ function is only interesting at
integer dimension, when $z=0$. Evaluated at $z=0$, \bas
\lim_{z\rightarrow 0} \beta_{\sigma_{dr}} (\varphi_{dr,u}(z)) \in
\g(\C)\eas is the $\beta$ function defined in \cite{CK01, CM06, EM,
  thesis1}. The limit is well defined because $\varphi_{dr,u}$ has
local counterterms.

For momentum cutoff regularization, the physical $\beta$ function is
\bas \beta_{\sigma_{mc}}(\varphi_{mc, z_0}(z(t), y(s))) =
\varphi^{-1}_{mc,z_0}(z_0, y_0) \star \dt \varphi_{mc,z_0}(z(t), y(s))
|_{t=1} \in \g(\C)\; .\eas It is a constant valued element of the Lie
algebra.

\subsection{A connection on the renormalization bundle}

In this section I define a connection on $B$ defined by the
logarithmic differential of sections $\varphi(z,y,t): B \rightarrow
P$.  Following \cite{CM06}, I call this connection
$\varphi^*\omega$. Following \cite{thesis1}, I show that these
connections are defined by $\beta_\sigma$, . The physical $\beta$
functions for a regularization scheme appear as pullbacks along
specific sections.  This brings me to the main theorem of the paper,
which I prove at the end of this section.

\begin{theorem}
The two physical $\beta$ functions
$\beta_{\sigma_{dr}}(\varphi_{dr,u}(z, y))$ and
$\beta_{\sigma_{mc}}(\varphi_{mc, z_0}(z(t),y(s))$ can be related by a
gauge transformation on the bundle $P \rightarrow B$. \end{theorem}

\begin{definition} Let $D$ be a differential operator.
\bas D:G(\scrA[t^{-1},t]) &\rightarrow
\Omega^1(\g(\scrA[t^{-1},t])) \\ \varphi (z, y,t) &\mapsto
\varphi(z,y,t)^{\star -1}\star d(\varphi( z,y,t))\;. \eas
\end{definition}

Let $\omega \in \Omega^1(P, \g)$ be a connection on $P$ defined
locally by the differential operator $D$ defined above.

\begin{lemma}For $f \in G(\scrA[t^{-1}, t])$, the differential
$D(f) = f^\star \omega$ defines a connection on
  section $f$ of $P \rightarrow B$.
\label{logmult}\end{lemma}

\begin{proof}
If $D$ defines a connection, it must satisfy equation \ba (f^{\star
  -1} \star g)^*\omega = g^{-1}dg + g^{\star -1}(f^*\omega) g \;
,\label{pullbackcond}\ea for $f,\, g \in G(\scrA[t^{-1},t])$. Since
$df^{-1} =- f^{-1} df f^{-1}$, \bas D(f^{-1} g) = Dg - g^{-1}f
f^{-1}df f^{-1}g \;, \eas or \bas Dg = D(f^{-1} g) + (f^{-1}g)^{-1}Df
(f^{-1}g) \;.\eas which satisfies equation
\eqref{pullbackcond}.
\end{proof}

The objects of interest in this paper are the vector fields
$\beta_\sigma \in TG(\scrA)$. To define these, I define a related connection
on the bundle $K_\sigma$.

\begin{definition}
Let $\omega_\sigma \in \Omega^1(K_\sigma, \g)$ be the connection on
$K_\sigma$ defined $\omega_\sigma := i_\sigma^*\omega$.
\end{definition}

Viewing $K_\sigma$ as a $G$ principal bundle over $B$, one can write
the corresponding connection over $B$ as \bas (\varphi,
t)^*\omega_\sigma(z, y, t) = \varphi_\sigma^{-1}d (\varphi_\sigma)\;.\eas As
an element in $\Omega^1(B, \g)$, I write it as \bas
(\varphi,t)^*\omega_\sigma(z, y, t) = a_\varphi(z,y,t) dz +
b_\varphi(z,y,t) dy + c_\varphi(z, y, t) dt \; , \eas with \ba
 a_\varphi(z,y,t) = \sigma_t\varphi(z,y)^{ \star -1} \star
 \frac{\partial}{\partial z} \sigma_t \varphi (z,y)
 \label{adef}\\ b_\varphi(z,y,t) = \sigma_t\varphi(z,y)^{ \star -1} \star
 \frac{\partial}{\partial y} \sigma_t \varphi (z,y)\label{bdef}
 \\ c_\varphi(z,y,t) = \sigma_t\varphi(z,y)^{ \star -1} \star \dt
 \sigma_t \varphi (z,y)
\label{cdef} \;.\ea

\begin{theorem}
The connection $\omega_\sigma$ is $\C^\times$ equivariant for any
renormalization group action $\sigma$.
\end{theorem}

\begin{proof}
This is a result of the bundle $K_\sigma$ being $\C^\times$
equivariant and that the derivative in $\C^\times$ is
exponential. Explicitly, one can check this on sections.

By equations \eqref{adef} and \eqref{bdef}, the coefficients
$a_\varphi$ and $b_\varphi$ involve derivatives with respect to $z$
and $y$ respectively. Therefore $\sigma_u(a_\varphi(z, y, t)) =
a_\varphi(z, y , tu) $ and $\sigma_u(b_\varphi(z, y, t) = b_\varphi(z,
y , tu))$. In the case of $c_\varphi$, \bas c_\varphi(z, y, tu) &=
\varphi_\sigma^{-1}(z, y,tu)\star t dt \varphi_\sigma(z, y,tu) \\ &=
\sigma_u\varphi_\sigma^{-1}(z, y, t) \star ut \partial_{tu}
\varphi_\sigma(z, y, tu)  \\&= \sigma_u(c_\varphi(z, y, t)) \eas
\end{proof}

By the $\C^\times$ equivariance, it is sufficient to study connections
of the form $\varphi_\sigma(z,y,1)^*\omega =
\varphi^*\omega_\sigma$. Notice that \bas c_\varphi(z, y, 1) =
\left[\sigma_t\varphi^{-1}(z, y) \star \dt
  \sigma_t\varphi(z, y)\right]|_{t=1} = \beta_\sigma(\varphi(z,y))\;
.\eas Furthermore, \ba \sigma_u\left(\beta_\sigma(\varphi(z,y))\right)
= \left[\sigma_t\varphi^{-1}(z, y) \star \dt
  \sigma_t\varphi(z, y)\right]|_{t=u} \;. \label{sigmabeta}\ea This,
combined with the fact that $\beta_\sigma$ is a bijection between
$G(\scrA)$ and $\g(A)$ gives the following:

\begin{theorem}
The connection $\omega_\sigma$ on $K_\sigma$ is defined by the vector field
generating the renormalization group action, $\beta_\sigma$.
\label{betadef}\end{theorem}

\begin{proof}
Let $\alpha(z, y) \in \g(\scrA)$, and $\rho_\sigma(1)$ the inverse
function of $\beta_\sigma$, as defined in Lemma \ref{betainv}.  Notice
that the quantities $a_\varphi(z,y,1)$, and $b_\varphi(z, y, 1)$ are
all functions of $\rho_\sigma(1)(c_\varphi(z,y,1))$.

Therefore, for any $\alpha \in \g(\scrA)$, I can write $\psi(z, y)
:=\rho_\sigma(1)(\alpha)\in G(\scrA)$. This defines a connection on
$B$ \bas \psi^*\omega_\sigma(z,y, 1) =
D(\rho_\sigma(1)(\alpha(z,y,1)))\;, \eas and by equation
\eqref{sigmabeta} \bas \psi^*\omega_\sigma(z,y, t) =
D(\rho_\sigma(t)(\alpha(z,y,t))) \;.\eas
\end{proof}

This allows one to relate two different physical $\beta$
functions.

\begin{theorem}
The physical $\beta$ functions
$\beta_{\sigma_{dr}}(\varphi_{dr,u}(z))$ and
$\beta_{\sigma_{mc}}(\varphi_{mc,z_0}(z, y))$ both define pullbacks of
a global connection, $\omega$ on $P$ to $B$. Therefore, they can be
related by a gauge transformation.
\end{theorem}
\begin{proof}
By theorem \eqref{betadef}, $\beta_{\sigma_{dr}}$ and
$\beta_{\sigma_{mc}}$ define the connections $\omega_{\sigma_{dr}}$
and $\omega_{\sigma_{mc}}$ on $K_{\sigma_{dr}}$ and $K_{\sigma_{mc}}$
respectively. These are both pullbacks of a connection $\omega$ defined
on $P$ by logarithmic differentiation.

Since $\omega_{\sigma_{dr}} = i_{\sigma_{dr}}\omega$, for
$i_{\sigma_{dr}}$ and inclusion map, \bas
\varphi_{dr,u}^*\omega_{\sigma_{dr}} = \varphi_{dr,u}^*\omega\eas and
\bas \varphi_{mc,z_0}^*\omega_{\sigma_{mc}} =
\varphi_{dr,u}^*\omega\eas are both connections on $B$ with
$\varphi_{dr,u} $ and $\varphi_{mc,z_0}$ both sections of
$P\rightarrow B$. The connections defined by
$\beta_{\sigma_{dr}}(\varphi_{dr,u}(z))$ and
$\beta_{\sigma_{mc}}(\varphi_{mc,z_0}(z(t),y(s)))$ can be related by the
gauge transformation \bas D(\varphi_{dr, u} \star \varphi_{mc, z_0}) =
D (\varphi_{mc,z_0})+ \varphi_{mc, z_0}^{\star -1}\star D(\varphi_{dr,
  u})\star \varphi_{mc,z_0} \; .\eas
\end{proof}

 Loop-wise calculations for the $\beta$ functions for dimensionally
 regularized and cut-off regularized quantum electrodynamics give
 different values, starting at the 3-loop level \cite{IZ}. This
 theorem gives a geometric structure for understanding the relation
 between the two renormalization schemes.  While this paper has
 specifically examined a sharp momentum cut-off regulator, there are
 other related regularization schemes, such as smooth cut-off or
 Pauli-Villars regularization, that also have a structure of
 logarithmic singularities and finite order poles. Theories under
 these regularization schemes, and their $\beta$ functions, can also
 be expressed in terms of sections and connections of this
 renormalization bundle.

\bibliographystyle{amsplain}
\bibliography{/home/mithu/bibliography/Bibliography}{}

\end{document}